\documentclass[conference]{IEEEtran} 
\usepackage{latexsym,amssymb,amsmath,graphicx,bbm}
\usepackage{epic}
\usepackage{psfrag}
\usepackage{amsbsy}
\usepackage{mathptm} 
\usepackage[bbgreekl]{mathbbol}
\usepackage{cmgreek} 
\usepackage{epsfig}
\usepackage{cite}
\usepackage[small,normal,bf,up]{caption2}

\newtheorem{lemma}{Lemma}[section]
\newtheorem{theorem}{Theorem}[section]
\newtheorem{corollary}{Corollary}[section]
\DeclareMathOperator{\prob}{{\text{\rm P}}}

\newcommand{\D}[2]{\frac{\partial #1}{\partial #2}}
\newcommand{\un}[1]{\underline{#1}}
\newcommand{\brc}[1]{\left({#1}\right)}

\newcommand{\abs}[1]{\left\lvert#1\right\rvert}

\newcommand{\openone}{\leavevmode\hbox{\small1\normalsize\kern-.33em1}}

\begin{document}
\renewcommand{\textfraction}{0}

\title{Bounds on Thresholds Related to Maximum Satisfiability of Regular Random Formulas}

\author{\IEEEauthorblockN{Vishwambhar Rathi\IEEEauthorrefmark{1}\IEEEauthorrefmark{2}, 
 Erik Aurell\IEEEauthorrefmark{1}\IEEEauthorrefmark{3},  
Lars Rasmussen\IEEEauthorrefmark{1}\IEEEauthorrefmark{2}, Mikael Skoglund\IEEEauthorrefmark{1}\IEEEauthorrefmark{2}\\
\{vish@, eaurell@, lars.rasmussen@ee, skoglund@ee\}.kth.se}
\IEEEauthorblockA{\IEEEauthorrefmark{1}School of Electrical Engineering}
\IEEEauthorblockA{\IEEEauthorrefmark{2}KTH Linnaeus Centre ACCESS}
\IEEEauthorblockA{KTH-Royal Institute of Technology, Stockholm, Sweden}
\IEEEauthorblockA{\IEEEauthorrefmark{3}Dept. Information and Computer Science, TKK-Helsinki University of Technology, 
Espoo, Finland}
}

\maketitle

\begin{abstract}
We consider the regular balanced model of formula generation in conjunctive
normal form (CNF) introduced by Boufkhad, Dubois, Interian, and Selman. We say
that a formula is $p$-satisfying if there is a truth assignment satisfying
$1-2^{-k}+p 2^{-k}$ fraction of clauses. Using the first moment method we
determine upper bound on the threshold clause density such that there are no
$p$-satisfying assignments with high probability above this upper bound.  There
are two aspects in deriving the lower bound using the second moment method. The
first aspect is, given any $p \in (0,1)$ and $k$, evaluate the lower bound on
the threshold. This evaluation is numerical in nature. The second
aspect is to derive the lower bound as a function of $p$ for large enough $k$.
We address the first aspect and evaluate the lower bound on the $p$-satisfying
threshold using the second moment method. We observe that as $k$ increases the
lower bound seems to converge to the asymptotically derived lower bound for
uniform model of formula generation by Achlioptas, Naor, and Peres. 
\end{abstract}

\section{Regular Formulas and Motivation}
A literal of a boolean variable is the variable itself or its negation. 
A clause is a disjunction (OR) of $k$ literals. A formula is a conjunction (AND)
of a finite set of clauses. A $k$-SAT formula is a formula where each clause is
a disjunction of $k$ literals.  A {\it legal} clause is one in which there
are no repeated or complementary literals. Using the terminology of
\cite{BDIS05}, we say that a formula is {\it simple} if it consists of only
legal clauses. A {\it configuration} formula is not necessarily legal. A
satisfying (SAT) assignment of a formula is a truth assignment of variables
for which the formula evaluates to true i.e.~all the clauses evaluates to true.  We
denote the number of variables by $n$, the number of clauses by $m$, and the clause
density i.e.  the ratio of clauses to variables by $\alpha = \frac{m}{n}$. 
We denote the binary entropy function by $h(x) \triangleq -x \ln(x) - (1-x) \ln(1-x)$, 
where the logarithm is the natural logarithm.
 
The popular, uniform k-SAT model generates a formula by selecting uniformly and
independently $m$-clauses from the set of all $2^k \binom{n}{k}$ $k$-clauses. In
this model, the literal degree can vary. We are interested in the model where
the literal degree is constant, which was introduced in
\cite{BDIS05}. Suppose each literal has degree $r$.  Then $2 n r = k m$, 
which gives $\alpha = 2 r / k$. Hence $\alpha$ can only take values from a 
discrete set of possible values.  This problem can be circumvented by allowing each literal to 
take two possible values for a degree. However, in this paper we consider the case 
where all  the literals have degree $r$ due to space restriction. 
A formula is represented by a bipartite graph. The left vertices represent the
literals and right vertices represent the clauses. A literal is connected to a
clause if it appears in the clause.  There are $k \alpha n$ edges coming out
from all the literals and $k \alpha n$ edges coming out from the clauses. We
assign the labels from the set $\mathcal{E}= \{1, \dots, k \alpha n\}$ to edges
on both sides of the bipartite graph. In order to generate a formula, we
generate a random permutation $\Pi$ on $\mathcal{E}$. Now we connect an edge
$i$ on the literal node side to an edge $\Pi(i)$ on the clause node side. This
gives rise to a regular random k-SAT formula. Note that not all the formulas
generated by this procedure are simple. However, it was shown in \cite{BDIS05}
that the threshold is  same for this collection of formulas and the
collection of simple formulas. Thus, we can work with the collection of
configuration formulas generated by this procedure. Note that this procedure is 
similar to the procedure of generating regular LDPC codes \cite{RiU08}. 


The regular random $k$-SAT formulas are of interest because such instances are
computationally harder than the uniform $k$-SAT instances. This was
experimentally observed in \cite{BDIS05}, where the authors also derived upper
and lower bounds on the satisfiability threshold for regular random $3$-SAT.
In \cite{RARS10}, upper bound on the satisfiability threshold for 
any $k \geq 3$ was derived using the first moment method.  It was shown that as
$k$ increases, the upper bound converge to the corresponding bounds on the
threshold of the uniform model \cite{AcP04, AcM06}. 

For uniform model, in a series of breakthrough papers, Achlioptas
and Peres in \cite{AcP04} and Achlioptas and Moore in \cite{AcM06}, derived
almost matching lower bounds with the upper bound by carefully applying the
second moment method to balanced satisfying assignments. In \cite{AcM06}, based
on their belief that the simple application of second moment method should work
for symmetric problems, Achlioptas and Moore posed the question of its success
 for regular random $k$-SAT. In an attempt to answer 
this question, the lower bound using the second moment method was evaluated in
\cite{RARS10}. As the evaluation of the lower bound was numerical in nature
(though exact), it was observed that the lower bound also converges to the
corresponding lower bound for the uniform model as $k$ increases.  

In this work we are interested in the maximum satisfiability problem over
regular random formulas. We say that a formula is $p$-satisfying if there is a 
truth assignment satisfying $c(p) \triangleq 1-2^{-k}+p 2^{-k}$ fraction of its
clauses, $p \in (0,1)$.  The number of $p$-satisfying truth assignments is
denoted by $N(n, \alpha, p)$.  We define the following quantities related to
$p$-satisfiability:  
\begin{eqnarray}
\alpha(p) & \triangleq & \sup\{\alpha: \text{A regular random formula}, \nonumber \\
&& \quad\quad\quad\text{is $p$-satisfiable w.h.p.}\}, \label{eq:defalpha}\\
\alpha^*(p) & \triangleq & \inf\{\alpha:\text{A regular random formula}, \nonumber \\
&& \quad\quad\quad\text{is not $p$-satisfiable w.h.p.}\}. \label{eq:defalphastar} 
\end{eqnarray}
Note that $\alpha(p) \leq \alpha^*(p)$. In \cite{ANP07}, for the uniform model
Achlioptas, Naor, and Peres derived lower bound on $\alpha(p)$ which almost
matches with the upper bound on $\alpha^*(p)$ derived via the first moment
method. The lower bound was obtained by a careful application of the second moment 
method. 

We derive upper bound on $\alpha^*(p)$ by applying the first moment method to
$N(n, \alpha, p)$.  The obtained upper bound matches with the corresponding
bound for the uniform model. We evaluate a lower bound on $\alpha(p)$ by
applying the second moment method to $N(n, \alpha, p)$. We observe that for
increasing $k$ the lower bound seems to converge to the corresponding bound for
the uniform model. In the next section, we obtain upper bound on $\alpha^*(p)$
by the first moment method. 

Due to space limitations, some of the arguments are accompanied by short explanations. 
Further details can be found in the forthcoming journal submission \cite{RARS10t}. 

\section{Upper Bound on Threshold via First Moment}
Let $X$ be a non-negative integer-valued random variable and $E(X)$ be its expectation. 
Then the first moment method gives: $\prob\brc{X > 0} \leq E(X)$. 
Note that by choosing $X$ to be the number of solutions of a random formula, we can 
obtain an upper bound on the threshold $\alpha^*(p)$ beyond which no $p$-satisfying solution exists 
with probability one. This upper bound corresponds to the largest value of $\alpha$ at which 
 the average number of $p$-satisfying solutions goes to zero as $n$ tends to infinity. 
In the following lemma, we derive the first moment of $N(n, \alpha, p)$ 
for  regular random $k$-SAT for $k \geq 2$.   
\begin{lemma}\label{lem:moment1}
Let $N(n, \alpha, p)$ be the number of $p$-satisfying assignments for a random regular k-SAT 
formula. Then\footnote{We assume that $k \alpha n$ is an even integer.}, 
\begin{eqnarray}
E(N(n, \alpha, p)) & = & 2^n \binom{\alpha n}{c(p) \alpha n} \frac{\brc{\brc{\frac{k \alpha n}{2}}!}^2}{(k \alpha n)!} \times 
\nonumber \\ 
&&  \mathrm{coef}\brc{\brc{\frac{s(x)}{x}}^{c(p) \alpha n}, x^{\brc{\frac{k}{2}-c(p)} \alpha n}} 
\label{eq:moment1},  
\end{eqnarray}
where $s(x) = (1+x)^k - 1$  and
$\mathrm{coef}\brc{\brc{\frac{s(x)}{x}}^{c(p) \alpha n}, x^{\brc{\frac{k}{2}-c(p)}
\alpha n}}$ denotes the coefficient of $x^{\brc{\frac{k}{2}-c(p)} \alpha n}$ in
the expansion of $\brc{\frac{s(x)}{x}}^{c(p) \alpha n}$. 
\end{lemma}
\begin{proof} 
Due to symmetry of the formula generation, any assignment of variables has the same 
probability of being $p$-satisfying. This implies
\[
E(N(n, \alpha, p)) = 2^n \prob\brc{X=\{0,\dots,0\}\text{ is }p\text{-satisfying.}}. 
\]
The probability of the all-zero vector being $p$-satisfying is given by 
\begin{multline*}
\prob\brc{X=\{0,\dots,0\}\text{ is }p\text{-satisfying}} = \\ \frac{\text{Number of formulas for which } X=\{0,\dots,0\} 
\text{ is }p\text{-satisfying}}{\text{Total number of formulas}}.
\end{multline*}
The total number of formulas is given by $(k \alpha n)!$. The total number of formulas for which the all-zero  
 assignment is a $p$-satisfying is given by
\[
 \binom{\alpha n}{c(p) \alpha n} \brc{\brc{\frac{k \alpha n}{2}}!}^2 \mathrm{coef}\brc{s(x)^{c(p) \alpha n}, x^{\frac{k \alpha n}{2}}}. 
\] 
The binomial term corresponds to choosing $c(p)$ fraction of clauses being satisfied by the all-zero 
assignment. The factorial terms correspond to permuting the edges among true and false literals. Note 
that there are equal numbers of true and false literals. The generating function $s(x)$ corresponds to 
placing at least one positive literal in a clause. With these results and observing that 
$\mathrm{coef}\brc{s(x)^{c(p) \alpha n}, x^{\frac{k \alpha n}{2}}} = 
\mathrm{coef}\brc{\brc{s(x)/x}^{c(p) \alpha n}, x^{\brc{\frac{k}{2}-c(p)}\alpha n}},$
we obtain (\ref{eq:moment1}).
\end{proof}

{\it Remark:} The computation of the first moment of $N(n, \alpha, p)$ is similar to 
the computation of the first moment for weight distribution of regular LDPC ensembles. 
There is large body of work dealing with first moment of weight distribution. So, we 
refer to \cite{RiU08} and the references there in. 

We now state the Hayman method to approximate the coef-term which is
asymptotically correct \cite{RiU08}. 
\begin{lemma}[Hayman Method]\label{lem:hayman} 
Let $q(y) = \sum_i q_i y^i$ be a polynomial with non-negative coefficients such that 
$q_0 \neq 0$ and $q_1 \neq 0$. Define 
\begin{equation}\label{eq:defaqbq}
   a_q(y) = y \frac{d q(y)}{d y} \frac{1}{y}, \quad b_q(y) = y \frac{d a_q(y)}{d y}. 
\end{equation}
Then, 
\begin{equation}\label{eq:hayman}
	\mathrm{coef}\brc{q(y)^n, y^{\omega n}} = \frac{q(y_\omega)^n}{(y_\omega)^{\omega n} 
\sqrt{2 \pi n b_q(y_\omega)}} (1+o(1)), 
\end{equation}
where $y_\omega$ is the unique positive solution of the saddle point equation $a_q(y) = \omega$. The 
solution $y_\omega$ also satisfies 
\begin{equation}\label{eq:yomegaasinf}
y_\omega = \inf_{y > 0} \frac{q(y)^n}{y^{\omega n}}. 
\end{equation}
\end{lemma}

We now use Lemma \ref{lem:hayman} to compute the expectation of the total number of 
$p$-satisfying assignments. 
\begin{lemma}\label{lem:moment1hayman}
Let $N(n, \alpha, p)$ denote the total number of $p$-satisfying 
 assignments of a regular random k-SAT 
formula. Let $t(x)=\frac{s(x)}{x}$,  where $s(x)$ is defined in Lemma \ref{lem:moment1}. Then, 
\begin{multline}\label{eq:moment1hayman}
E(N(n, \alpha, p)) = \sqrt{\frac{k}{8 \pi c(p)^2 (1-c(p)) b_t(x_k) \alpha n}}  
 e^{n \brc{ (1 - k \alpha) \ln(2)}} \\
 e^{n \brc{\alpha h(c(p)) + c(p) \alpha \ln\brc{t(x_k)}-\brc{\frac{k \alpha}{2}-\alpha c(p)} 
\ln\brc{x_k}}} (1+o(1)), 
\end{multline}
where $x_k$ is the solution of $a_t(x) = \frac{k}{2 c(p)}-1$. The quantity $a_t(x)$ and  
 $b_t(x)$ are defined according to (\ref{eq:defaqbq}).
\end{lemma}
\begin{proof}
Using Stirling's approximation for the binomial terms (see \cite[p. 513]{RiU08})
 and Hayman approximation for the $\text{coef}$ term 
from Lemma \ref{eq:hayman}  gives the desired result. 
\end{proof}

In the following lemma we derive explicit upper bounds on the clause density for the existence of 
$p$-satisfying assignments.
\begin{lemma}[Upper bound]\label{lem:regupperbound}
Let $\alpha^*(p)$ be as defined in (\ref{eq:defalphastar}). Define
$\alpha_u^*(p)$ to be the upper bound on $\alpha^*(p)$ obtained by the first moment method. Then,
\begin{equation}\label{eq:ubalpha}
\alpha^*(p) \leq  \alpha_u^*(p) \triangleq \frac{2^k \ln(2)}{p + (1-p) \ln(1-p)}.
\end{equation}
\end{lemma}
\begin{proof}
Using (\ref{eq:yomegaasinf}), we obtain the following upper bound on the exponent of  $E(N(n, \alpha, p)))$ for any 
$x > 0$, 
\begin{multline}\label{eq:ubENx} 
\lim_{n \to \infty} \frac{\ln\brc{E(N(n, \alpha))}}{n} \leq  
(1 - k \alpha) \ln(2) + \alpha h(c(p)) + \\
c(p) \alpha \ln(t(x)) - \alpha \brc{\frac{k}{2}-c(p)} \ln(x).
\end{multline}
We substitute $x=1$ in (\ref{eq:ubENx}) and obtain the following upper bound 
on the $p$-satisfiability threshold, 
\begin{equation}\label{eq:temp4}
\alpha^*(p) \leq \frac{\ln(2)}{k \ln(2)-h(c(p))-c(p) \ln\brc{2^k-1}}. 
\end{equation}
We complete the proof by showing that the denominator in (\ref{eq:temp4}) is
lower bounded by $2^{-k} \brc{p + (1-p) \ln(1-p)}$. This can be easily shown by
considering their difference and then showing its positivity.  
\end{proof}

Note that the upper bound coincides with the upper bound derived for uniform formulas in 
\cite{ANP07}. In the next section we use the second moment method to obtain lower bound 
on $\alpha(p)$ defined in (\ref{eq:defalpha}). 
\section{Second Moment}
In \cite{ANP07}, almost matching lower bounds on the $p$-satisfiability threshold of 
uniform formulas were derived using the second moment method.
 The second moment method is governed by the following equation
\begin{equation}\label{eq:smmethod}
  \prob\brc{X > 0} \geq \frac{E(X)^2}{E(X^2)}, 
\end{equation}
where $X$ is a non-negative random variable. Before we use the second moment 
method, we present the following theorem from \cite{BFU93} and its corollary 
given in \cite{ANP07}. Proof of both theorem and the corollary are identical 
for the uniform model and the regular model.  
\begin{theorem}[\cite{BFU93}]
Let $U_k(n, \alpha)$ be the maximum number of clauses satisfied (over all the truth assignments) 
of a regular random formula with $n$ variables and $m$ clauses. Then
\[
\prob\brc{\left|U_k(n, \alpha) - E\brc{U_k(n, \alpha)}\right| > t} < 2 \exp\brc{-\frac{2 t^2}{\alpha n}}.
\]
\end{theorem}
\begin{corollary}[\cite{ANP07}]\label{cor:problower}
Assume that there exists $c=c(k, p, r)$ such that for $n$ large enough, a regular random 
formula is $p$-satisfiable with probability greater than $n^{-c}$. Then a regular random 
formula is $p'$-satisfiable with high probability for every constant $p' < p$. 
\end{corollary}

We now apply the second moment method to $N(n, \alpha, p)$. The calculation for $p=1$ i.e. for number of
satisfying assignments was done in \cite{RARS10}.  Our computation of the second
moment is inspired by the computation of the second moment for the weight and
stopping set distributions of regular LDPC codes in \cite{Vra06, Vra08} (see also \cite{BaB05}). We compute
the second moment in the next lemma. 
\begin{lemma}\label{lem:moment2}
Consider regular random satisfiability formulas with literal degree $r$.  
 Then the second moment of $N(n, \alpha, p)$ is given by 
\begin{multline}\label{eq:moment2}
E\brc{N(n, \alpha, p)^2} = \sum_{i=0}^n \sum_{j=\alpha n \brc{c(p)-(1-p) 2^{-k}}}^{c(p) \alpha n} 2^n \binom{n}{i}  
\binom{\alpha n}{c(p) \alpha n} \\
\binom{c(p) \alpha n}{j} \binom{\alpha n (1-c(p))}{c(p) \alpha n-j} 
 \frac{\brc{\brc{r(n -i )}!}^2 \brc{(r i)!}^2}{(k \alpha n)!} \\ 
\mathrm{coef}\brc{f(x_1, x_2, x_3)^j (s(x_1) s(x_3))^{c(p) \alpha n-j}, (x_1 x_3)^{r (n-i)} x_2^{r i}}, 
\end{multline}
where the generating function $f(x_1, x_2, x_3)$ is given by
\begin{equation}\label{eq:deff}
f(x_1, x_2, x_3) = (1+x_1+x_2+x_3)^k - (1+x_1)^k - (1+x_3)^k + 1.
\end{equation}
The generating function $s(x)$ is defined as $s(x) \triangleq (1+x)^k-1$, which is same as defined in 
Lemma \ref{lem:moment1}.
%
\end{lemma}
\begin{proof}
For truth assignments $X$ and $Y$, define the indicator variable 
 $\openone_{X Y}$ which evaluates to $1$ if 
the truth assignments $X$ and $Y$ are $p$-satisfying. Then,  
\begin{eqnarray*}
E(N(n, \alpha, p)^2) & = & \sum_{X, Y \in \{0, 1\}^n} E\brc{\openone_{\bf{X} \bf {Y}}}, \\
&=& 2^n \sum_{Y \in \{0, 1\}^n} \prob\brc{{\bf 0} \text{ and } Y \text{ are } p\text{-satisfying}}. 
\end{eqnarray*}
Due to the symmetry in regular formula generation, the number of formulas for
which both $X$ and $Y$ are $p$-satisfying depends only on the number of variables 
on which $X$ and $Y$ agree. This explains the last simplification where  we
fix $X$ to be the all-zero vector.  

We want to evaluate the probability of the event that the truth assignments ${\bf 0}$ and $Y$ 
 $p$-satisfy a randomly chosen regular 
formula. This probability  depends only on the {\it overlap}, i.e., the number
of variables where the two truth assignments agree. Thus for a given overlap
$i$, we can fix $Y$ to be equal to zero in the first $i$ variables and equal to
$1$ in the remaining variables i.e. 
$Y=\{\underbrace{0,\cdots,0}_{i \text{ times}}, \underbrace{1,\cdots,1}_{n-i \text{ times}}\}$. 
This gives,
\begin{equation}\label{eq:m2temp1}
E(N(n, \alpha, p)^2) = \sum_{i=0}^{n} 2^n \binom{n}{i} \prob\brc{{\bf 0} \text{ and } Y \text{ are } 
p\text{-satisfying}}.
\end{equation}
In order to evaluate the probability that both ${\bf 0}$ and $Y$ are $p$-satisfying, 
define $\mathcal{C}=\{1,\cdots,\alpha n\}$ to be the set of clauses and  
$C_{\bf 0}$ and $C_Y$ to be the set of clauses satisfied by ${\bf 0}$ and $Y$ respectively. 
Clearly, $|C_{\bf 0}|=|C_Y|=c(p) \alpha n$. Then, 
\begin{multline}
\prob\brc{{\bf 0} \text{ and } Y \text{ are } p\text{-satisfying}} = \\
\sum_{C_{\bf 0}, C_Y \subset \mathcal{C}} \prob\brc{{\bf 0} \text{ only satisfies } C_{\bf 0} \text{ and }
Y \text{ only satisfies } C_Y}. 
\end{multline}
Again from the symmetry of the regular formula generation, we fix 
$C_{\bf 0}=\{1,\cdots,c(p) \alpha n\}$. For $|C_{\bf 0} \cap C_Y| = j$, 
we fix $C_Y=\{1, \cdots, j, c(p) \alpha n+1,\dots, 2 c(p) \alpha n-j\}$. This gives, 
\begin{multline}\label{eq:m2temp2}
 \prob\brc{{\bf 0} \text{ and } Y \text{ are } p\text{-satisfying}} = 
 \binom{\alpha n}{c(p) \alpha n} \\ 
\sum_{j=\alpha n \brc{c(p)-(1-p) 2^{-k}}}^{c(p) \alpha n}
\binom{c(p) \alpha n}{j} \binom{\alpha n (1-c(p))}{c(p) \alpha n-j} \\ 
\times \prob\brc{{\bf 0} \text{ only satisfies } C_{\bf 0} \text{ and } Y \text{ only satisfies } C_Y}. 
\end{multline}
Note that $j \geq \alpha n (c(p)-(1-p) 2^{-k})$ as $|C_{\bf 0}| + |C_{Y}| - |C_{\bf 0} \cap C_Y| \leq \alpha n$.  For a given overlap $i$ between ${\bf 0}$ and $Y$, 
we observe that there are four different types of edges 
connecting the literals and the clauses.  There are $r (n-i)$ {\bf type 1}
edges which are connected to true literals w.r.t. the {\bf 0} truth assignment and 
false w.r.t. to the $Y$ truth assignment. The $r i$ {\bf type 2} edges are connected to
true literals w.r.t. both the  truth assignments.  There are $r (n-i)$ {\bf
type 3} edges which are connected to false literals w.r.t. the {\bf 0} truth assignment and true
literals w.r.t. to the $Y$ truth assignment.  The $r i$ {\bf type 4} edges are connected to
false literals w.r.t. both the  truth assignments. Let $f(x_1, x_2, x_3)$ be
the generating function counting the number of possible edge connections to a
clause such that the clause is satisfied by both ${\bf 0}$ and $Y$. In $f(x_1, x_2, x_3)$, 
the power of $x_i$ gives the number of edges of type $i$, $i \in \{1,
2, 3\}$. A clause is satisfied by both ${\bf 0}$ and $Y$ if it is connected to at least one type $2$ edge,
else it is connected to at least one type $1$
and at least one type $3$ edge. Then the generating function $f(x_1, x_2, x_3)$
is given as in (\ref{eq:deff}).  Using this, we obtain
\begin{multline}\label{eq:m2temp3}
\prob\brc{{\bf 0} \text{ only satisfies } C_{\bf 0} \text{ and } Y \text{ only satisfies } C_Y} = \\ 
\frac{\brc{(r (n-i))!}^2 \brc{(r i)!}^2} {(k \alpha n)!} \times \\  
\mathrm{coef}\brc{f(x_1, x_2, x_3)^j (s(x_1) s(x_3))^{c(p) \alpha n-j}, (x_1 x_3)^{r (n-i)} x_2^{r i}},
\end{multline}
where $s(x_1)$ is the generating function for clauses satisfied by ${\bf 0}$ and not satisfied by $Y$ 
(similarly we define $s(x_3)$). The term $(k \alpha n)!$ is the total number of formulas. Consider a 
given formula which is satisfied by both  truth 
assignments ${\bf 0}$ and $Y$. If we permute the positions of type $1$ edges on the clause side, we obtain another 
formula having ${\bf 0}$ and $Y$ as solutions. The argument holds true for the type $i$ edges, $i \in \{2, 3, 4\}$.
This explains the term $(r (n-i))!$ in (\ref{eq:m2temp3}) which corresponds to permuting the type $1$ edges (it is squared because of the 
same contribution from type $3$ edges). Similarly, $(r i)!^2$ corresponds to permuting type $2$ and type $4$ edges. As $|C_{\bf 0} \cap C_Y|=j$, there are $j$ clauses which satisfied by both ${\bf 0}$ and 
$Y$. This explains the factor $f(x_1,x_2,x_3)^j$ in the coef term. There are 
$|C_{\bf 0} \backslash (C_{\bf 0} \cup C_Y)|)=
|C_{Y} \backslash (C_{\bf 0} \cup C_Y)|=c(p)\alpha n -j$ clauses which are satisfied by ${\bf 0}$ 
(resp. $Y$) and not satisfied by $Y$ (resp. {\bf 0}). This explains the factor 
$(s(x_1) s(x_3))^{c(p) \alpha n-j}$ in the coef term. We complete the proof by substituting 
(\ref{eq:m2temp3}) in (\ref{eq:m2temp2}), then (\ref{eq:m2temp2}) in (\ref{eq:m2temp1}). 
%
%
%
\end{proof}

In order to evaluate the second moment, we now present the multidimensional saddle point method 
in the next lemma. A detailed technical exposition of the multidimensional saddle point method 
can be found in Appendix D of \cite{RiU08}.  
\begin{theorem} \label{thm:multisaddle}
Let $\un{i}:=(r (n-i), r i, r (n-i))$,  $\un{x} = (x_1, x_2, x_3)$, and $
0 < \lim_{n \to \infty} i/n < 1$. We define 
\[
 g_{n, j}(\un{x}) = f(\un{x})^j (s(x_1) s(x_3))^{c(p) \alpha n-j}, 
\]
where $f(\un{x})$ and $s(x)$ are defined in Lemma \ref{lem:moment2}. We define the normalizations 
$\eta \triangleq i/n$ and $\gamma \triangleq j/(\alpha n)$.  
 Let $\un{t}=(t_1,t_2,t_3)$ be a positive solution of the saddle point equations  
\begin{equation}
a_g(\un{x}) \triangleq \left\{ \frac{x_i}{n} \frac{\partial \ln\brc{g_{n, j}(\un{x})}}{\partial x_i} \right\}_{i=1}^3 = 
\{r (1-\eta), r \eta, r (1-\eta)\}. 
\end{equation}
Then $\mathrm{coef}\brc{g_{n, j}(\un{x}),\un{x}^{\un{i}}}$ can be approximated as ,
\[
\mathrm{coef}\brc{g_{n, j}(\un{x}),\un{x}^{\un{i}}} = \frac{g_{n, j}(\un{t})}{({\un{t}})^{\un{i}}\sqrt{(2 \pi n)^3 \abs{B_g(\un{t})}}}(1+o(1)),
\]
using the saddle point method for multivariate polynomials, 
where $B_g(\un{x})$ is a $3 \times 3$ matrix whose elements are given by 
$B_{i,j}=x_j \D{a_{gi}(\un{x})}{x_j}=B_{j,i}$ and $a_{g i}(\un{x})$ is the $i^\text{\tiny th}$ coordinate of $a_g(\un{x})$.
\end{theorem} 

In the following theorem we derive lower bound on the $p$-satisfiability threshold by evaluating the second moment of 
$N(n, \alpha, p)$
with the help of Theorem \ref{thm:multisaddle}. 
\begin{theorem}\label{thm:reglowerbound}
Consider regular random $k$-SAT formulas with literal degree $r$. Let $S(i, j)$ denote the 
$(i, j)^{\text{th}}$ summation term in (\ref{eq:moment2}). Define the normalization $\eta=i/n$ and 
$\gamma=j/(\alpha n)$. If $S(n/2, n \alpha c(p)^2)$ is the dominant term i.e.~for  
$\eta \in [0, 1], \gamma \in [(c(p)-2^{-k} (1-p)), c(p)], \eta \neq \frac{1}{2}$, and 
$\gamma \neq c(p)^2$ 
\begin{equation}\label{eq:snbytwodominant}
 \lim_{n \to \infty} \frac{\ln\brc{S\brc{\frac{n}{2}, n \alpha c(p)^2 }}}{n} > 
\lim_{n \to \infty} \frac{\ln\brc{S(\eta n, \gamma \alpha n)}}{n}, 
\end{equation}
then for some positive constants $c$, $c'$ 
\begin{equation}\label{eq:positiveprobbound}
 \prob\brc{N(n, \alpha, p) > 0} \geq  c' n^{-c}.  
\end{equation}
Let $r^*(p)$ be the largest literal degree for which $S(n/2, n \alpha c(p)^2)$ is the dominant term, i.e. 
(\ref{eq:snbytwodominant}) holds. Then the threshold $\alpha(p)$ defined in (\ref{eq:defalpha}) is lower bounded by 
$\alpha(p) \geq \alpha_l^*(p) \triangleq \frac {2 r^*(p)}{k}$. 
\end{theorem}
\begin{proof}
Assuming (\ref{eq:snbytwodominant}) holds, then for $n$ large enough
\begin{equation}\label{eq:moment2ub}
E\brc{N(n, \alpha, p)^2} \leq (n+1) (\alpha (1-p) 2^{-k} n + 1) S\brc{\frac{n}{2}, n \alpha c(p)^2}. 
\end{equation}
From (\ref{eq:moment2}) and Theorem \ref{thm:multisaddle}, the growth rate of $S(\eta n, \gamma \alpha n)$ is given by, 
\begin{multline}\label{eq:defsgamma}
s(\eta, \gamma) \triangleq \lim_{n \to \infty} \frac{\ln\brc{S(\eta n, \gamma \alpha n)}}{n} = (1-k \alpha) (\ln(2) + h(\eta)) \\ 
+ \alpha h(c(p)) + \alpha c(p) h\brc{\frac{\gamma}{c(p)}} + \alpha (1-c(p)) h\brc{\frac{c(p)-\gamma}{(1-c(p))}} \\ 
+ \gamma \alpha \ln\brc{f(t_1, t_2, t_3)} 
+ \alpha (c(p)-\gamma) \brc{\ln(s(t_1)) + \ln(s(t_3))} \\ - r (1-\eta) (\ln(t_1)+\ln(t_3)) -r \eta \ln(t_2),  
\end{multline}
where $t_1, t_2, t_3$ is a positive solution of the saddle point equations as defined in 
Theorem \ref{thm:multisaddle}, 
\begin{equation}\label{eq:saddlepoint}
a_g(\un{t}) = \left\{r (1-\eta), r\eta, r (1-\eta)\right\}. 
\end{equation}
In order to compute the maximum exponent of the summation terms, we compute its partial derivatives with respect to $\eta$ 
and $\gamma$ and equate them to zero, which result in the following respective equations. 
\begin{eqnarray}
(1-k \alpha) \ln\frac{1-\eta}{\eta} + r \ln\brc{\frac{t_1 t_3}{t_2}} & = & 0, \label{eq:etaderivative}\\
\ln\brc{\frac{(c(p)-\gamma)^2}{\gamma (1-2 c(p)+\gamma)}} + \ln\brc{\frac{f(t_1, t_2, t_3)}{s(t_1) s(t_3)}} 
& = & 0. \label{eq:gammaderivative}
\end{eqnarray}
Note that the partial derivatives of $t_1, t_2$ and $t_3$ w.r.t. $\eta$ and $\gamma$ vanish because of the saddle point equations given in 
(\ref{eq:saddlepoint}). 
Every positive solution $(t_1, t_2, t_3)$ of (\ref{eq:saddlepoint}) satisfies 
$t_1=t_3$  as (\ref{eq:saddlepoint}) and $f(t_1, t_2, t_3)$ are symmetric in 
$t_1$ and $t_3$. If $\eta = 1/2, \gamma=c(p)^2$ is a maximum, then the vanishing derivative 
in (\ref{eq:etaderivative}) and equality of 
$t_1$ and $t_3$ imply $t_2 = t_1^2$. We substitute $\eta=1/2$, $t_1=t_3$, and 
$t_2=t_1^2$ in (\ref{eq:saddlepoint}). This reduces (\ref{eq:saddlepoint}) to the saddle point equation corresponding to 
the polynomial $t(x)$ defined in Lemma \ref{lem:moment1hayman} whose solution is 
denote by $x_k$. By observing $f(x_k, x_k^2, x_k) = s(x_k)^2$, we have 
\begin{multline}\label{eq:Snbytwo}
\! \! \! \! \! \!S\brc{\frac{n}{2}, \alpha c(p)^2 n} = 
\frac{k^{3/2}}{32 \pi^2 c(p)^2 (1-c(p))^2 \sqrt{|B_g(x_k, x_k^2, x_k)|} n^2} \times \\
 e^{2 n \brc{(1-k \alpha) \ln(2) + \alpha h(c(p)) + \alpha c(p) \ln(s(x_k)) - \frac{k \alpha}{2} \ln(x_k)}} (1+o(1)). 
\end{multline}
Using the relation that $t(x) = \frac{s(x)}{x}$,  we note that the exponent of $S(n/2, \alpha c(p)^2 n)$ is 
twice the exponent of the first moment of the total number of solutions as given in (\ref{eq:moment1hayman}). 
We substitute (\ref{eq:Snbytwo}) in to (\ref{eq:moment2ub}). Then we use Lemma \ref{lem:moment1hayman} and 
(\ref{eq:moment2ub}) in the second moment 
method:
\begin{eqnarray*}
\prob\brc{N(n, \alpha, p) > 0}  & \geq & \frac{E(N(n, \alpha, p)^2)}{E(N(n, \alpha, p)^2)}, \\ 
&\geq& \frac{4 \pi \sqrt{|B_g(x_k, x_k^2, x_k)|}}{b_t(x_k) \alpha^2 \sqrt{k} n} (1+o(1)).
\end{eqnarray*}
Clearly, if the maximum of the growth rate of $S(\eta n, \gamma \alpha n)$ is not
achieved at $\eta = 1/2$ and $\gamma=c(p)^2$ , then the lower bound
given by the second moment method converges to zero exponentially fast.  
By using Corollary \ref{cor:problower}, we obtain the desired lower bound on $\alpha(p)$. 
  This proves the theorem. 
\end{proof}

In the next section we discuss the obtained lower and upper bounds on the $p$-satisfiability 
threshold. 
\section{Bounds on Threshold and Discussion}
In Table \ref{tb:regthreshold}, we compute the ratio of $\alpha_l^\star(p)$ and
$\alpha_u^\star(p)$ for $p=0.1,\cdots,0.9$ and $k=3,6,12$, where
$\alpha_l^\star(p)$ is the lower bound on $\alpha(p)$ obtained from Theorem
\ref{thm:reglowerbound} and $\alpha_u^*(p)$ is the upper bound on $\alpha^*(p)$
defined in Lemma \ref{lem:regupperbound}. Note that the case $p=1$ was already 
solved in \cite{RARS10}. 
\begin{table}[h]
\begin{center}
\begin{tabular}{|c|c|c|c|}
\hline 
$p$ & $k=3$ & $k=6$ & $k=12$ \\
\hline
$0.1$ & $0.252$ & $0.717$ & $0.977$ \\
\hline
$0.2$ & $0.258$ & $0.720$ & $0.979$ \\
\hline
$0.3$ & $0.272$ & $0.738$ & $0.980$ \\
\hline
$0.4$ & $0.281$ & $0.755$ & $0.981$ \\
\hline
$0.5$ & $0.295$ & $0.765$ & $0.983$ \\
\hline
$0.6$ & $0.308$ & $0.782$ & $0.986$ \\
\hline
$0.7$ & $0.325$ & $0.801$ & $0.988$ \\
\hline
$0.8$ & $0.344$ & $0.822$ & $0.990$ \\
\hline
$0.9$ & $0.402$ & $0.855$ & $0.993$ \\
\hline
\end{tabular} 
\end{center}
\caption{Value of the ratio $\alpha_l^*(p)/\alpha_u^*(p).$ 
}\label{tb:regthreshold} 
\end{table}
In order to apply the second moment method, we have to verify that $s(\eta, \gamma)$,
defined in (\ref{eq:defsgamma}), attains its maximum at $\eta=\frac{1}{2}, \gamma=c(p)^2$ 
over $[0, 1] \times [c(p) - (1-p) 2^{-k}, c(p)]$.  This requires that 
$\eta=\frac{1}{2}, \gamma= c(p)^2$ is the only
positive solution of the system of equations consisting of
(\ref{eq:saddlepoint}), (\ref{eq:etaderivative}) and
(\ref{eq:gammaderivative}) which corresponds to a maximum. The system of equations (\ref{eq:saddlepoint}),
(\ref{eq:etaderivative}), and (\ref{eq:gammaderivative}) is equivalent to a
system of polynomial equations.  For small values of $k$ and $p$ close to one, we can solve this
system of polynomial equations and verify the desired conditions. For larger 
values of $k$, the degree of monomials in (\ref{eq:etaderivative}) grows
exponentially in $k$. Thus, solving (\ref{eq:saddlepoint}),
(\ref{eq:etaderivative}), and (\ref{eq:gammaderivative}) becomes
computationally difficult. However, $s(\eta, \gamma)$ can be easily computed as its
computation requires solving only (\ref{eq:saddlepoint}), where the maximum
monomial degree is linear in $k$. Thus, the desired condition for maximum of
$s(\eta, \gamma)$ at $\eta=\frac{1}{2}, \gamma=c(p)^2$ can be verified numerically in an efficient
manner.  

From the Table \ref{tb:regthreshold}, we see that as $k$ becomes larger the ratio 
$\alpha_l^*(p)/\alpha_u^*(p)$ gets closer to one. Our belief is that indeed as 
$k$ becomes larger and larger, the ratio $\alpha_l^*(p)/\alpha_u^*(p)$ converges to one. 
Our main future goal is to derive explicit expression for $\alpha_l^*(p)$ as $k$ becomes 
larger. 

\section*{Acknowledgement}
This work has been supported by Swedish research council (VR) through KTH Linnaeus center ACCESS. 

\bibliographystyle{siam}
\bibliography{ksat} 

\newcommand{\SortNoop}[1]{}
\begin{thebibliography}{10}

\bibitem{AcM06}
{\sc D.~AChlioptas and C.~Moore}, {\em Random k-{SAT}: Two moments suffice to
  cross a sharp threshold}, {SIAM} J. {COMPUT.}, 36 (2006), pp.~740--762.

\bibitem{ANP07}
{\sc D.~Achlioptas, A.~Naor, and Y.~Peres}, {\em On the maximum satisfiability
  of random formulas}, Journal of the Association of Computing Machinary
  (JACM), 54 (2007).

\bibitem{AcP04}
{\sc D.~AChlioptas and Y.~Peres}, {\em The threshold for random {$k$-SAT} is
  {$2^k \ln(2) -O(k)$}}, Journal of the American Mathematical Society, 17
  (2004), pp.~947--973.

\bibitem{BaB05}
{\sc O.~Barak and D.~Burshtein}, {\em Lower bounds on the spectrum and error
  rate of {LDPC} code ensembles}, in International Symposium on Information
  Theory, Adelaide, Australia, 2005.

\bibitem{BDIS05}
{\sc Y.~Boufkhad, O.~Dubois, Y.~Interian, and B.~Selman}, {\em Regular random
  k-{SAT}: Properties of balanced formulas}, Journal of Automated Reasoning,
  (2005).

\bibitem{BFU93}
{\sc A.~Z. Broder, A.~M. Frieze, and E.~Upfal}, {\em On the satisfiability and
  maximum satisfiability of random 3-{CNF} formulas}, Proc. 4th Annual
  Symposium on Discrete Algorithms,  (1993), pp.~322--330.

\bibitem{Vra06}
{\sc V.~Rathi}, {\em On the asymptotic weight and stopping set distributions of
  regular {LDPC ensembles}}, IEEE Trans. Inform. Theory, 52 (2006),
  pp.~4212--4218.

\bibitem{Vra08}
\leavevmode\vrule height 2pt depth -1.6pt width 23pt, {\em Non-binary LDPC
  codes and EXIT like functions}, PhD thesis, Swiss Federal Institute of
  Technology (EPFL), Lausanne, 2008.

\bibitem{RARS10}
{\sc V.~Rathi, E.~Aurell, L.~Rasmussen, and M.~Skoglund}, {\em Bounds on
  thresholds of regular random $k$-{SAT}}.
\newblock Accepted to International Conference on Theory and Applications of
  Satisfiability Testing (SAT) 2010.

\bibitem{RARS10t}
\leavevmode\vrule height 2pt depth -1.6pt width 23pt, {\em Satisfiability and
  maximum satisfiability of regular random $k$-sat: Bounds on thresholds}.
\newblock in preparation for submission to IEEE Transactions on Information
  Theory.

\bibitem{RiU08}
{\sc T.~Richardson and R.~Urbanke}, {\em Modern Coding Theory}, Cambridge
  University Press, 2008.

\end{thebibliography}
\end{document}